\documentclass[a4paper,11pt]{amsart}

\pdfoutput=1 
\usepackage{amsthm,amsmath,amssymb,amsfonts,mathtools}   
\usepackage{slashed,latexsym,multirow} 
\usepackage{graphicx,fullpage,xcolor}
\usepackage[makeroom]{cancel}
\usepackage[caption = false]{subfig}
\usepackage{enumitem,float}
\usepackage{physics, verbatim}

\usepackage{hyperref}
\hypersetup{colorlinks=true, linkcolor=blue, citecolor=red, urlcolor=blue}

\newtheorem{theorem}{Theorem}[section] 
\newtheorem{corollary}{Corollary}[theorem] 
\newtheorem{lemma}[theorem]{Lemma} 
\newtheorem{proposition}[theorem]{Proposition} 
\newtheorem{remark}[theorem]{Remark}

\newcommand{\ds}{{\slashed\partial}}
\newcommand{\A}{{\mathcal A}}
\newcommand{\B}{{\mathcal B}}
\newcommand{\T}{{\mathcal T}}

\newcommand{\HH}{{\mathcal H}}
\newcommand{\M}{{\mathcal M}}

\newcommand{\C}{{\mathbb C}}

\newcommand{\cinf}{{C^\infty(\M)}}
\newcommand{\I}{\mathbb I}

\begin{document} 
\title{{\bf Minimal twist of almost commutative geometries}}
\author{\vspace{-.3truecm} Manuele Filaci\textsuperscript{\dag},  Pierre Martinetti\textsuperscript{*\dag}}
\maketitle

\vspace{-.5truecm}\begin{center}
  \textsuperscript{*}\emph{Universit\`a di Genova, Dpt di
    Matematica},
  \textsuperscript{\dag}\emph{INFN sezione di Genova,} \\[0pt]
  \emph{via Dodecaneso, 16146 Genova GE Italy.} \\[6pt]
  \emph{E-mail:} manuele.filaci@ge.infn.it, martinetti@dima.unige.it,
\end{center}
\smallskip

\begin{center}
  \emph{In memoriam John Madore}
\end{center}

  \begin{abstract}
    We classify the  twists of almost
    commutative spectral triples that keep the Hilbert space and the Dirac
    operator untouched. The involved twisting operator is shown to be
    the product of the grading of a manifold by a finite dimensional
    operator, which is not necessarily a grading of the internal space. Necessary and sufficient constraints on this operator
    are listed.
  \end{abstract}

\section{Introduction}

The Higgs field might be a probe of the internal structure of spacetime. This idea - pioneered in
\cite{Dubois-Violette:1989fk} - has been fully implemented in the
framework of  noncommutative geometry \cite{Connes:1994kx}. It  yields a
description of the Standard Model of fundamental interactions (including massive neutrinos~\cite{Chamseddine:2007oz})  as a
pure gravity theory on an \emph{almost commutative geometry}
\cite{Connes:1996fu}, that is on 
a space which is the product of a usual continuous manifold $\M$ by some
internal matricial structure. 

This
product is mathematically well defined in terms of spectral
triple. The latter consists in an involutive algebra $\A$ acting
faithfully on some Hilbert space $\HH$, together with a selfadjoint
operator $D$ with compact resolvent, such that the commutator
$[D, a]$
is bounded for any $a$ in~$\A$. With additional axioms, spectral
triples furnish a purely algebraic characterization of
riemannian (spin) manifolds \cite{connesreconstruct}, as well as their generalisation to the
noncommutative setting~\cite{Connes:1996fu}. 

A twisted spectral triple is defined similarly, except that the
commutator is no longer required to be bounded. Instead, one
asks for an automorphism $\rho$ of $\A$ such that  the twisted commutator 
\begin{equation}
\label{eq:7b}
  [D, a]_\rho := D a - \rho(a) D
\end{equation}
is bounded for any $a$ in $\A$. Such twists have been introduced in \cite{Connes:1938fk} with some mathematical
motivations. Later, they show to be useful for  physical
applications as well, for they offer a way to build models beyond the
Standard Model \cite{Devastato:2013fk,buckley}. In particular, by twisting the spectral triple of the
Standard Model in a minimal way, that is keeping the Hilbert
space and Dirac operator untouched (only the algebra is modified), one produces extra bosonic fields without
altering the fermionic content of the theory.

There exists a  general procedure to obtain such minimal twist,  recalled in
section \ref{sec:minimal-twist-1}, which  uses a
grading of the spectral triple \cite{Lett.}. Recall that the later is a selfadjoint operator
$\Gamma$ on $\HH$, squaring to the identity $\I$  and anticommuting with
$D$, such that $[\Gamma, a]=0$ for any $a\in\A$. 
 The aim of this note is to understand  which of these properties
are necessary:
  is the twist doable using a \emph{twisting operator}
  that is not a grading ?

  We first list in section
\ref{sec:twisting-operator}  some basic properties
expected from the twisting operator. Then we proceed at the light of three conditions that must be satisfied by a real
twisted spectral triple: the boundedness of the twisted commutator
\eqref{eq:7b} in section \ref{sec:bound-comm-1}, the order zero condition
in section \ref{sec:order-zero-condition} and the twisted first-order
condition in section \ref{sec:twisted-first-order}. 
The resulting constraints  are listed in propositions 
 \ref{sec:order-zero-condition-1} and
 \ref{sec:twisted-first-order-2}. The grading
operator  is not the only solution.

For 
almost-commutative geometries, assuming the
twisting operator is the product of an
operator $\mathcal T$ acting on the spinor space of $\M$ with
an operator $T_F$ acting on the internal space, these constraints are
shown to be equivalent to their reduction to the internal space
(corollaries  \ref{sec:order-zero-condition-2} and  \ref{sec:twisted-first-order-1})
. Although the boundedness of the commutator forces 
$\mathcal T$ to be the grading  of the manifold, 
$T_F$  is not necessarily
a grading of the internal space.

\newpage
\section{Minimal twist}
\label{sec:minimal-twist-1}

A \emph{minimal twist} of a (real, graded) spectral triple  $(\A, \HH,
D)$ is \cite[Def. 3.2]{Lett.} 
a (real, graded) twisted spectral triple $(\A\otimes{\mathcal B}, \HH,
D)_\rho$ where $\mathcal B$ is a involutive algebra with unit $1_\B$,
$\rho$ an automorphism of $\A\otimes{\mathcal B}$ and the
representation $\pi$ of the latter on $\HH$  is such that 
\begin{equation}
  \label{eq:4bis}
  \pi(a\otimes 1_\B) =\pi_0(a) \quad \forall a\in \A,
\end{equation}
where $\pi_0$ is the representation of $\A$ on $\HH$ from the initial spectral
triple.

If this initial spectral triple is graded, then
there always exists a
minimal twist with{\footnote{To fix
notation  we assume that $\A$ is a complex algebra, but the results
 also hold for real algebras.}}  $\B= \C^2$
and $\rho$ the automorphism of $\A\otimes \C^2\simeq \A\oplus\A$ given by the flip
\begin{equation}
  \label{eq:4}
\rho((a, a')):=  (a', a) \quad \forall a, a'\in \A\otimes\C^2.
\end{equation}
The construction of this minimal \emph{twist-by-grading},
\begin{equation}
 \label{eq:11}
(\A\otimes\C^2, \HH,
D)_\rho,
\end{equation}
starts with  the following
observation:   by definition the grading
$\Gamma$ commutes with the algebra, so the projections $\frac{\I\pm \Gamma}2$ on its
eigenspaces $\HH_\pm$  define two independent involutive
representations
\begin{equation}
\label{eq:5}
   \pi_\pm(a):= \left(\frac{\I\pm \Gamma}2\pi_0(a)\right)_{\HH_\pm}
\end{equation}
 of $\A$ on~$\HH_\pm$. Their direct sum
\begin{equation}
\label{eq:3}
  \pi(a,a'):=\frac{\I+ \Gamma}2\pi_0(a) +
  \frac{\I- \Gamma}2\pi_0(a')\qquad \forall a, a'\in\A
\end{equation} 
is a representation of $\A\otimes\C^2$ that satisfies the properties
of a twisted spectral triple \cite[Prop.3.8]{Lett.} as well as
condition \eqref{eq:4bis}.

 This twist-by-grading is  the only
possible minimal twist for the spectral triple naturally associated to
an (even dimensional) closed riemannian spin manifold $\M$ \cite[Prop.4.2]{Lett.}, namely
\begin{equation}
\label{eq:13}
  \cinf,\quad L^2(\M, S),\quad \ds = -i\sum_{\mu=1}^{\text{dim} \M} \gamma^\mu\nabla_\mu
\end{equation}
where the unital algebra $\cinf$ of smooth functions on $\M$ acts by
multiplication on the
Hilbert space $L^2(\M, S)$  of square integrable spinors,
\begin{equation}
\label{eq:1}
  (\pi_\M(f)\psi)(x):=f(x)\psi(x) \quad \forall\psi\in L^2(\M,S), x\in \M,
\end{equation}
and $\ds$ is the Dirac operator
associated with the spin structure. This spectral triple is graded
with grading the product $\gamma_{\M}$ of the gamma matrices.

The unicity of this twist  no longer holds true for an almost commutative geometry 
\begin{equation}
\label{eq:16}
 \A = \cinf\otimes \A_F,\quad \HH= L^2(\M, S)\otimes \HH_F,\quad D=
  \ds\otimes\I_F + \gamma_\M\otimes D_F,
\end{equation} 
that is the product of \eqref{eq:13} with a finite dimensional graded
spectral
triple $(\A_F, \HH_F, D_F)$ with grading $\Gamma_F$ (in the equation
above $\I_F$ is the identity operator on $\HH_F$). The
representation 
\begin{equation}
\label{eq:2}
  \pi_0 = \pi_\M\otimes \pi_F
\end{equation}
of $\A$ on $\HH$ is the product of $\pi_\M$ with the
representation $\pi_F$ of $\A_F$ on $\HH_F$ given by the finite dimensional
spectral triple. If $\pi_F$ is irreducible, then any minimal twist is
necessarily by $\B=\C^2$ but the representation $\pi$  is not necessarily
the one given in
\eqref{eq:3}, as explained  below. If  $\pi_F$  is not irreducible, there exists minimal twists with $\B$ different from
$\C^2$ \cite[Corr.4.5]{Lett.}. 

\begin{remark}
  The dimension of $\M$ has to be even so that the spectral triple
  \eqref{eq:13} admits a  grading $\gamma_\M$. The odd dimensional case should be
  investigated elsewhere.
\end{remark}

\section{Twisting operator}
\label{sec:twisting-operator}

The point of this note is to investigate which  properties of
the grading $\Gamma$ are necessary to build a minimally twisted partner
to a usual spectral triple. The commutativity with the
initial representation of $\A$ is important to get  two
independent representations $\pi_\pm$, but to what extend are the commutation properties of $\Gamma$
with $D$ and  (in case of a real spectraal triple) with the real structure $J$ relevant  ? 

Given a spectral triple $(\A, \HH,D)$, we thus consider 
an operator in $\mathcal B(\mathcal
H)$,
which shares all the properties of a grading but the commutation
properties with $D$ and~$J$. Namely $T$ is selfadjoint, $T^2=\I$, the degeneracy of both its eigenvalues
$\pm 1$ are non-zero and  $T$ commutes with the
representation $\pi_0$ of $\A$ on $\HH$.  
 The latter is thus  
the direct sum $\pi_+ \oplus\pi_-$ of the two involutive representations of $\A$ on the eigenspaces $\HH_\pm$ of $T$ given by
 \begin{equation}
   \pi_\pm(a)=\left(\frac{\I\pm  T}2\pi_0(a)\right)_{\HH_\pm} .
 \end{equation}
As in \eqref{eq:3}, the operator  $T$ allows to define a
representation of $\A\otimes\C^2$  on $\HH$
\begin{equation}
\label{eq:3bis}
  \pi(a,a'):=    \pi_+(a)\oplus    \pi_-(a')=\frac{\I+ T}2\pi_0(a) + \frac{\I- T}2\pi_0(a').
\end{equation}
To avoid  domain issues, we assume that $T\HH \subset
\text{Dom } D$.  We call $T$ a \emph{twisting operator}.

For an almost commutative geometry \eqref{eq:16}, we further assume that \begin{equation}
\label{eq:7}
 T= {\mathcal T}\otimes T_F
 \end{equation}
where ${\mathcal T}\in L^2(\M,S)$ and $T_F\in {\mathcal B}(\HH_F)$. A bounded operator on $\HH$ is not necessarily of this form, but could be (the closure of) a sum of such
operators. However, we restrict to operators \eqref{eq:7}, for they already pave the way to interesting physical
applications beyond the Standard Model.
The selfadjointness of $T$ is to guarantee that the representations
$\pi_\pm$ are involutive. This does not imply that $\T$ and $T_F$
are selfadjoint. However one may
always
restrict to this case.
\begin{lemma}
  \label{sec:minimal-twist}
Let $T=\T \otimes T_F$ be a selfadjoint operator on $L^2(\M
,S)\otimes \HH_F$ that squares to~$\I$. Then there exist two selfadjoint operators $\tilde\T$ on $L^2(\M,S)$ and
$\tilde T_F$ on $\HH_F$, squaring to the identity, such that 
  $T=\tilde\T \otimes\tilde T_F$.
\end{lemma}
\begin{proof}
  The matrix $T_F^\dag T_F$ being non-zero (otherwise $T$ does not
  square to $\I$) is positive, thus it admits at least one real eigenvalue $\lambda>0$, with associated
  eigenvectors $\psi\in\HH_F$, and all the other non-zero eigenvalues are
  also strictly positive.  For any $\varphi\in
  L^2(\M, S)$, one has 
  \begin{equation}
T^\dagger T (\varphi\otimes\psi) =  \T^\dagger\T \varphi\otimes T_F^\dagger T_F
\psi =\lambda \T^\dagger\T \varphi\otimes\psi.
  \end{equation}
On the other side, by hypothesis $T^\dag T=\I$, that is $T^\dagger T (\varphi\otimes\psi)  =
\varphi\otimes\psi$. Therefore
\begin{equation}
(\lambda\T^\dagger\T - \I_\M) \varphi\otimes\psi = 0 \qquad
\forall \varphi\in L^2(M, S)\, 
\end{equation}
meaning that $\T^\dagger\T$ coincides with the operator of
multiplication of spinors by $\lambda^{-1}$. Repeating the analysis for another
non-zero eigenvalue $\lambda'$ shows that $\T^\dagger\T$ coincides
with the multiplication by ${\lambda'}^{-1}$, so
$\lambda'=\lambda$. This allows to define
\begin{equation}
\label{eq:15}
\tilde \T =\lambda^{-\frac 12}\T,\quad
\tilde T_F= \lambda^{\frac 12} T_F,
\end{equation}
such that $T=\tilde\T\otimes \tilde T_F$
with
\begin{equation}
\tilde\T^\dagger \tilde\T =\lambda^{-\frac 12}\T^\dagger \lambda^{-\frac 12}\T=\lambda^{-1}\T^\dagger\T=\I_\M.
\end{equation}
From $T^\dagger T = T T^\dagger=\I$ then follows that $\tilde
T_F^\dagger \tilde T_F =
\tilde T_F \tilde T_F^\dagger=\I_F$, that is $\tilde T_F$ is unitary,
and so is $T_F$.

To show that $\tilde \T$ and $\tilde T_F$ are selfadjoint, let us
apply $T=T^\dag$ on $\varphi\otimes\Psi$ where $\Psi$ is  an
eigenvector of  $\tilde T_F$, with eigenvalue $\tau\in\C$,
$|\tau|=1$. Using $\tilde T_F^\dag\Psi =\tau^{-1}\Psi$,  one obtains
\begin{equation}
  \tilde\T\varphi\otimes \tau \Psi=\tilde\T^\dag\varphi\otimes
  \tau^{-1}\Psi\quad\forall\varphi\in L^2(\M ,S),
\end{equation}
meaning that 
  $\tau^{-1} \tilde\T^\dagger =\tau \tilde\T$.
Redefining $\tau\tilde\T\rightarrow \tilde\T $ (that is 
$\tau^{-1}\tilde \T^\dagger\rightarrow \tilde\T^\dagger$), the
previous~equation 
shows that $\tilde T$ is selfadjoint. The selfadjointness of $\tilde T_F$
then follows from the one~of~$T$.
\end{proof}
\newpage

\section{Boundedness of the commutator}
\label{sec:bound-comm-1}
 
We investigate the conditions
imposed, on the twisting operator \eqref{eq:7} of an almost commutative
geometry, by the boundedness of the twisted commutator
\begin{equation}
\label{eq:14}
  [\ds\otimes\I_F + \gamma^5\otimes D_F, \pi(a,a')]_\rho,
\end{equation}
for $\pi$ the representation \eqref{eq:3bis} and  $\rho$ the flip
\eqref{eq:4}. For simplicity, we restrict to the case $\pi_F$ is irreducible, that is
${\mathcal B}=\C^2$. 

\begin{lemma}
\label{prop:bound-comm} 
For an almost commutative geometry, the twisted
  commutator $[\ds\otimes\I_F, \pi(a, a')]_\rho$ is bounded for any $(a, a')\in
  \A\otimes\C^2$ if and only if $\T$ anticommutes with $\ds$.
\end{lemma}
\begin{proof}
 Using
 \begin{align}
   \label{eq:10}
(\ds\otimes \I_F)(\I\pm T)=(\I\mp T) (\ds\otimes \I_F)\pm\left\{\ds\otimes
   \I_F,T\right\},
 \end{align}
one obtains, omitting the symbol $\pi_0$,
\begin{align}
\label{eq:8}
 [\ds\otimes \I_F, \pi(a, a')]_\rho&= (\ds\otimes \I_F) \pi(a,a')-\pi(a',a)(\ds\otimes \I_F),\\
 &= (\ds\otimes \I_F) \left(\frac{\I+ T}2a + \frac{\I- T}2a'\right) - \left(\frac{\I+ T}2a' + \frac{\I- T}2a\right)
 (\ds\otimes \I_F),\\
\label{eq:12}
&=\frac{\I- T}2 [\ds\otimes \I_F,  a] + \frac{\I+ T}2  [\ds\otimes
  \I_F , a'] +\frac 12\left\{\ds\otimes\I_F,T\right\} (a-a').
\end{align}
For $a=f\otimes m$, then $[\ds\otimes \I_F,  a]= [\ds, f]\otimes m$
is  bounded, being  $(\A,
\HH, D)$ a spectral triple. The same is true for an arbitrary $a$ in
$\A$,  and also for $[\ds\otimes \I_F,  a']$. So the first two terms in
\eqref{eq:12} are bounded.

If $T$ anticommutes with
$\ds\otimes\I_F$, the last term in \eqref{eq:12}  is zero, so that \eqref{eq:8} is
bounded.

Conversely,  assume \eqref{eq:8} is bounded for any  $(a,a')$ in
$\A\otimes \C^2$. This means that the last term in \eqref{eq:12} is
bounded. For $a-a'= 1\otimes m$ with $1$ the constant function
$f(x)=1$ on $\M$, then this last term is (up to a factor $\frac 12$)
\begin{equation}
  \label{eq:9}
  \left\{\ds\otimes\I_F,T\right\} (a-a')= \left\{\ds,\T \right\}
  \otimes T_F m.
\end{equation}
This is bounded if and only if $ \left\{\ds,\T \right\}$ is
bounded. For $\psi$ on  $\HH_+$ ($+1$ eigenspace of $\T$), one has 
\begin{equation}
  \left\{\ds,\T \right\}\psi = \ds\psi + \T\ds\psi = (\I+\T)\ds\psi
\end{equation}
meaning  $\left\{\ds,\T \right\}$ coincides with  $(\I+\T)\ds$ which is
an unbounded operator, unless it is zero. So the restriction of
$\left\{\ds,\T \right\}$ to $\HH_+$ is zero. A similar argument holds
for the restriction to $\HH_-$. Hence the result.
\end{proof}

The finite part of an almost commutative geometry only involves
bounded operator. 
Therefore the boundedness of the twisted commutator \eqref{eq:14} only
depends on the property of $\T$.
\begin{proposition}
\label{sec:bound-comm-2}
 The twisted commutator \eqref{eq:14} is bounded if, and only if, 
 $\T=\pm\gamma_\M$.
\end{proposition}
\begin{proof}
The twisted commutator $[\gamma^5\otimes D_F, \pi(a, a']_\rho$ is
bounded, whether or not $T$ anticommutes with
$\gamma^5\otimes D_F$. So \eqref{eq:14} is bounded iff
$[\ds\otimes\I_2, \pi(a, a')]_\rho$ is bounded, that is by
Prop.~\ref{prop:bound-comm} iff $\T$
anticommutes with $\ds$. Explicitly, with $\nabla_\mu
=\partial_\mu+\omega_\mu$ where $\omega_\mu$ is the spin connection,
this~means
\begin{equation}
\label{eq:31}
  \left\{-i\gamma^\mu\partial_\mu,\T\right\}+\left\{-i\gamma^\mu\omega_\mu,\T\right\}=0.
\end{equation}
The second term is bounded, being $\omega_\mu$ bounded. By the Leibniz rule
satisfied by $\ds$, one has
\begin{align}
   \left\{-i\gamma^\mu\partial_\mu,\T\right\}&=
   -i\gamma^\mu\partial_\mu\T -i\T\gamma^\mu\partial_\mu = 
   -i\gamma^\mu(\partial_\mu\T)   -i\gamma^\mu \T \partial_\mu 
   -i\gamma^\mu\partial_\mu\T ,\\
&=  -i\gamma^\mu(\partial_\mu\T)   -i\left\{\gamma^\mu ,\T \right\}\partial_\mu. 
\end{align}
The first term is bounded, the second one unbounded. For \eqref{eq:31}
to hold, both the bounded part and the unbounded parts must be
zero. The latter condition is equivalent to $\left\{\gamma^\mu ,\T
\right\}=0$, so $\T=\lambda\gamma_\M$, for the only operator that
anti-commutes with all the $\gamma$'s matrices are the multiple of
$\gamma_\M$. 
By lemma \ref{sec:minimal-twist} $\T$ is selfadjoint -  which forces
 $\lambda$ to be real - and
$\T^2=\I_\M$,  which reduces the choice to $\lambda=\pm 1$. 
\end{proof}

\section{order-zero condition}
\label{sec:order-zero-condition}

Not all the axioms of noncommutative geometry have been adapted to the
twisted context. However, those most relevant for physics
(i.e. regarding gauge transformations) do make sense for a twisted
spectral triple \cite{Lett.,Landi:2017aa}. Especially, 
a real structure for a twisted spectral triple is defined as in the
non-twisted case, that is an antilinear operator $J$ such that 
\begin{equation}
J^2=\epsilon \I,\; JD = \epsilon' DJ , \; J\Gamma = \epsilon''\Gamma J
\end{equation}
for some $\epsilon, \epsilon',\epsilon''\in\left\{-1,1\right\}$ (those
  three signs defines the $KO$-dimension of the triple), which implements a representation $\pi^\circ$ of the 
opposite algebra $\A^\circ$,
\begin{equation}
\pi_0^\circ(a^\circ) = J\pi_0(a^*)J^{-1}
\end{equation}
asking to 
commutes with the one of $\A$,
\begin{equation}
\label{eq:6}
[\pi_0(a),   J\pi_0(b^*)J^{-1}]=0 \qquad \forall a, b\in \A.
\end{equation}
This is the order zero condition (the first-order
condition, discussed in the next section).

The
twist-by-grading of a real twisted spectral triple $(\A, \HH,D)$ automatically
satisfies the order zero condition, with the same real structure. Namely, if \eqref{eq:6} holds for $(\A, \HH, D)$, then
for $\pi$ the representation  \eqref{eq:3} of $\A\otimes\C^2$ defined
by the grading one has
\begin{equation}
\label{eq:20}
  [\pi(a,a'),    J\pi(b^*, {b'}^*)J^{-1}]=0 \qquad \forall\,  (a, a'), (b,b')\in \A\otimes\C^2.
\end{equation}
We work out below the
conditions such that the same holds true for the representation $\pi$
\eqref{eq:3bis} induced by the twisting operator.
\begin{proposition}
\label{sec:order-zero-condition-1}
  Let $(\A, \HH, D)$ be a real spectral triple with real structure
  $J$. If $\A$ is a unital algebra, then the order zero condition \eqref{eq:20} for the
  representation $\pi$ in \eqref{eq:3bis} holds true if and only if
  \begin{equation}
\label{eq:27}
 [T, JTJ^{-1}]=0     \quad \text{ and } \quad  [\pi_0(a), JTJ^{-1}]=0\qquad
    \forall a \in\A.
  \end{equation}
\end{proposition}
\begin{proof}
By easy manipulations, one has
  \begin{align}
\nonumber   [\pi(a,a'), J\pi(b, b')J^{-1}]&= \left[\frac12\left(\I+T\right) \pi_0(a) + \frac
                                    12\left(\I-T\right)\pi_0(a'), J\left(\frac
                                    12\left(\I+T\right)\pi_0(b) +
                                    \frac
                                    12\left(\I-T\right)\pi_0(b')\right)J^{-1}\right],\\
 \label{eq:21}
    &= \frac 12\left[\pi_0(\alpha) +T\pi_0(\alpha'),J\left(\pi_0(\beta) +T\pi_0(\beta')\right)J^{-1}\right],
  \end{align}
  where we write
  \begin{equation}
\label{eq:35}
    \alpha:= a+ a',\; \alpha'=a-a',\; \beta:=b+b',\; \beta'=b-b'.
  \end{equation}
The order zero condition is equivalent to \eqref{eq:21} being
  zero for any $\alpha, \alpha', \beta, \beta'\in \A$, that is (omitting the
  symbol of representation and denoting with an hat the adjoint action
  of $J$, e.g. $\hat T=JTJ^{-1}$)
  \begin{align}
\label{eq:22}
    &\left[\alpha, \hat \beta \right]=0,\quad\;\,\,
      \left[\alpha, JT\beta J^{-1}\right]=0,\\
    \label{eq:23}&\left[T\alpha, \hat\beta\right]=0,\quad\left[T\alpha, JT\beta
      J^{-1}\right]=0\qquad \forall \alpha,\beta\in\A.
  \end{align}

The first condition \eqref{eq:22} is the order zero condition for
$(\A, \HH,D)$, so it is always true by hypothesis. The second
condition \eqref{eq:22} writes
\begin{equation}
  0= \left[\alpha, \hat T\, \hat\beta
   \right]=\hat T\left[\alpha,  \hat\beta \right]+
  \left[\alpha, \hat T\right]\hat\beta.
\end{equation}
The first term is zero by the order zero condition. The second term is zero for any $\alpha, \beta$ if and only
if (consider the case $\beta$ is the unit of $\A$) 
\begin{equation}
\label{eq:24}
  \left[\alpha, \hat T\right] \quad \forall \alpha\in \A.
\end{equation}

For the same reasons, the first equation \eqref{eq:23} written as
\begin{equation}
0=  T[\alpha, \hat \beta ] + [T,\hat\beta ]\alpha 
\end{equation}
is equivalent to $[T, \hat \beta ]=0$ that is, multiplying by
$J^{-1}$ on the left and $J$
on the right,  
\begin{equation}
0=J^{-1} TJ\beta- \beta J^{-1}T J= [J^{-1}T J, \beta].
\end{equation}
Remembering that $J^{-1}=\epsilon J$, this is equivalent to
\eqref{eq:24}.  Remains the second condition \eqref{eq:23}, which
is equivalent to
\begin{equation}
\label{eq:25}
  \hat  T [T\alpha,  \hat\beta ] + [T\alpha, \hat T ] \hat \beta =0,
\end{equation}
It implies (consider $\alpha=\beta=1_\A$)
\begin{equation}
\label{eq:26}
  [T, \hat T]=0.
\end{equation}
Therefore  the order zero condition implies \eqref{eq:24}
and \eqref{eq:26}. To show that these two conditions are sufficient, the only points that remains to show is that they
imply the second condition \eqref{eq:23}, that is \eqref{eq:25}. The
first term of this equation is zero as soon as  \eqref{eq:24} holds
(as shown studying the first term of \eqref{eq:23} ). The second term
is obviously zero as soon as both $T$ and $\alpha$ commute with $\hat T$.
\end{proof}

For an almost commutative geometry 
the conditions \eqref{eq:27} are equivalent to their
restrictions to the finite dimensional spectral triple.
To show this, let us absorb the sign ambiguity of  proposition \ref{sec:bound-comm-2} 
redefining $T_F\to \pm T_F$. The  twisting operator
\eqref{eq:7} of an almost commutative geometry is thus
\begin{equation}
\label{eq:33}
  T=\gamma_\M\otimes T_F.
\end{equation}
Since $\gamma_\M$ commutes with the representation $\pi_\M$ \eqref{eq:1}, 
requiring $T$ to commute with the representation $\pi_0$ \eqref{eq:2}
implies 
\begin{equation}
\label{eq:34}
  [T_F, \pi_F(m)]=0\quad \forall m\in \A_F.
\end{equation}
\begin{corollary}
  \label{sec:order-zero-condition-2}
The minimal twist of a real, almost commutative, geometry by $T$ as in
\eqref{eq:33}  satisfies the order zero condition
if and only if
\begin{equation}
  \label{eq:30}
[T_F, J_F T_F J_F^{-1}]=0 
\quad\text{ and }\quad[T_F, J_F \,\pi_F(m) \, J_F^{-1}]=0 \quad \forall
 m\in \A_F.
\end{equation}
\end{corollary}
\begin{proof}
Since ${\mathcal J}\gamma_\M=\epsilon''\gamma_\M{\mathcal J}$, then 
\begin{equation}
  JTJ^{-1}= {\mathcal J}\gamma_M{\mathcal J}^{-1}\otimes J_F T_F
  J_F^{-1} = \epsilon'' \gamma_\M\otimes J_F T_FJ_F^{-1}.
\end{equation}
The conditions \eqref{eq:27} then become
  \begin{equation}
    \epsilon''\I\otimes[T_F, J_F T_F J_F^{-1}]=0, \quad
    \epsilon''\gamma_M f\otimes [m, J_F T_F J_F^{-1}]=0 \quad \forall
    f\otimes m\in \A.
    \label{eq:29}
  \end{equation}
  This is equivalent to \eqref{eq:30} 
  using that
  $[m, J_F T_F J_F^{-1}]=0$ is equivalent to
  $[T_F, J_F m J_F^{-1}]~=~0$.
\end{proof}

Since $T_F$ commutes with the representation of $\A_F$ by \eqref{eq:34}, the two
conditions \eqref{eq:30} are automatically satisfied if $T_F$ commutes or
anticommutes with $J_F$, as required for a grading operator. However, by the
order zero condition of the initial triple, these
conditions together with \eqref{eq:34} are also satisfied if $T_F$ is the representation of any
element in the center of  $\A_F$, not necessarily a grading. Furthermore these are not the only
 possibilities: think for instance of a spectral triple obtained by
taking a subalgebra $\A'_F$ of a spectral triple $(\A_F, \HH_F,
D_F)$. Then any elements in the center of $\A_F$ can be taken as
$T_F$, even if it is not in~$\A'_F$.

\section{Twisted first order condition}
\label{sec:twisted-first-order}

The operator $T_F$ is further constrained if one takes into
account the twisted version of the first order condition:
\begin{equation}
\label{eq:17}
 [ [D, \pi((a,a'))]_\rho,\, J\pi_0((b^*,{b'}^*))J^{-1}]_{\rho^\circ} =0 \quad
 \forall (a,a'),\,  (b,b')\in\A\otimes \C^2
\end{equation}
where $\rho^\circ$ is the automorphism of the opposite algebra
$(\A\otimes\C^2)^\circ\simeq \A^\circ\oplus\A^\circ$ induced by $\rho$:
\begin{equation}
\label{eq:18}
 \rho^\circ((a^\circ, {a'}^\circ)) := (\rho^{-1}((a,a'))^\circ=
 (a', a)^\circ=({a'}^\circ, a^\circ).
\end{equation}

\begin{proposition}
\label{sec:twisted-first-order-2}
  Consider the minimal twist by $T$ of a real spectral triple, such that the order zero condition holds. Then the twisted first-order condition 
 holds if and only if
\begin{equation}
  \label{eq:30fo}
\left\{\left\{D, T\right\}, J T J^{-1}\right\}=0 \quad\text{ and }\quad
[\left\{D,T\right\}, J \,\pi_0(a) \, J^{-1}]=0 \quad \forall
 a\in \A.
\end{equation}
\end{proposition}
\begin{proof}
For the representation \eqref{eq:3bis} and omitting the symbol
of representation $\pi_0$, one has
\begin{equation*}
  [D,\pi(a,a')]_\rho= \frac 12 D\left((\I+T)a + (\I-T)a' \right)- \frac 12\left((\I+T)a' + (\I-T)a \right)D= \frac 12 [D, a+a'] + \frac 12\left\{D, T(a-a')\right\},
\end{equation*}
while, denoting $\hat b:=JbJ^{-1}$ the conjugation by $J$,
\begin{equation}
  J\pi(b,b')J^{-1}= \frac12 J\left( (\I+T)b +
    (\I-T)b'\right)J^{-1}=\frac12(\hat b +\widehat{b'})+ J T(b-b')J^{-1}.
\end{equation}
So the twisted commutator \eqref{eq:17} (with $b,b'$ instead of $b^*,
{b'}^*)$ is the sum of a term of order
$0$~in~$T$,
\begin{equation}
 \frac 14\left[ [D,a+a'], \hat b +\widehat{b'}\right]
\end{equation}
 a term of order $1$,
 \begin{equation}
\label{eq:19}
    \frac 14\left\{ [D,a+a'],  J
      T(b-b')J^{-1}\right\}+  \frac 14\left[ \left\{D,
        T(a-a')\right\},  \hat b +\hat b'\right],
 \end{equation}
and a term of order $2$,
\begin{equation}
   \frac 14\left\{ \left\{D,
        T(a-a')\right\},  J
      T(b-b')J^{-1}
       \right\}. 
\end{equation}

The term of order $0$ is always zero by the first order condition of the
initial triple. The first component of the term of order $1$ must
vanish independently: indeed, for $a=a'$, the second term in~\eqref{eq:19} as well as the term of order $2$ are zero, so the twisted
first-order condition reduces~to
\begin{equation}
  \label{eq:36}
\left\{ [D,\alpha],  \hat T
  \hat\beta'\right\}=0 \quad \forall \alpha=a+a',\; \beta'=b-b'.
\end{equation}
By the twisted
first-order condition of the initial triple, one has
\begin{align}
  \label{eq:37}
\left\{ [D,\alpha],  \hat T
  \hat\beta'\right\}&=  [D,\alpha]  \hat T
  \hat\beta' + \hat T
  \hat\beta' [D,\alpha] = [D,\alpha]  \hat T
  \hat\beta' + \hat T
 [D,\alpha] \hat\beta' ,\\
\label{eq:42}
&=\left\{D,\hat T\right\}\alpha\hat\beta' - \alpha\left\{D,\hat T\right\}\hat\beta'
\end{align}
where the second equation follows developing the commutators, and
using that $\hat T$ commutes with $\alpha$ by proposition
\ref{sec:order-zero-condition-1}. In particular, for $\beta=1_\A$  one gets that \eqref{eq:36}
- hence the twisted first-order condition -  implies
\begin{equation}
\label{eq:41}
 \left[ \left\{D,\hat T\right\},\alpha\right]=0 \quad \forall \alpha\in \A.
\end{equation}

As well, the term of order $2$ must vanish independently: for $a=-a'$,
$b'=-1_\A$, $b=-2b'$ then the term of order $1$ vanishes, so the
twisted first-order condition reduces to
\begin{equation}
   \left\{ \left\{D,
        T\alpha'\right\}, 
      \hat T
       \right\} = 0 \quad \forall \alpha'=a-a'\in\A. 
\end{equation}
In particular, for $\alpha'=1_\A$, one gets that the twisted
first-order condition implies
\begin{equation}
\label{eq:45}
\left\{\left\{D,T\right\},
    \hat T\right\}=0.
\end{equation}

Therefore \eqref{eq:41} and \eqref{eq:45} are necessary to
get the twisted first-order condition. Let us show  they are
sufficient conditions. If \eqref{eq:41} holds, then \eqref{eq:42}
vanishes for any $\alpha, \hat\beta'$, meaning the first component of the
term of order $1$ vanishes for any $a,a',b,b'$. The same is true for
the second component since, with $\beta=\hat b+\hat b'$, the later writes
\begin{align}
  \label{eq:43}
\left[\left\{D,T\alpha'\right\},\hat\beta\right]&=\left[ [D,\alpha']T , \hat\beta\right]+ \left[\alpha'\left\{D,T\right\} ,\hat\beta\right]
\end{align}
where we use
\begin{equation}
  \label{eq:44}
\left\{D,T\alpha'\right\}= [D,\alpha']T + \alpha'\left\{D,T\right\}
\end{equation}
obtained by direct computation, with $T$ commuting with
$\alpha'$ by definition of twisting operator. The first
commutator in \eqref{eq:43} vanishes because $\hat \beta$ commutes with
both $[D, \alpha']$ (by the first order condition of the initial
triple) and with $T$ (by the second equation \eqref{eq:27} rewritten as $[T,\hat
a]=0$ for all $a\in\A$). The second commutator in \eqref{eq:43}
vanishes as well since $\bar\beta$ commutes with both $\alpha'$ (by
the order zero condition of the initial triple) and with
$\left\{D,T\right\}$ by \eqref{eq:41} rewritten as
\begin{equation}
\label{eq:46b}
\left[\left\{D,T\right\}, \hat\alpha\right]=0\quad \forall a\in\A.
\end{equation}
Finally, by \eqref{eq:44}  the term of order $2$ writes
\begin{equation}
\label{eq:39}
\left\{ \left\{D,   T\alpha'\right\},  \hat T \hat\beta'
\right\}= \left\{[D,\alpha']T , \hat T \hat\beta'
\right\}+ \left\{\alpha'\left\{D,T\right\}, \hat T \hat\beta'
\right\}.
\end{equation}
The second anti-commutator vanishes, for $\hat T\hat\beta'$
anticommutes with $\left\{D,T\right\}$ (by \eqref{eq:45} and
\eqref{eq:46b}) but commutes with $\alpha$ (by \eqref{eq:27} and the
order zero condition). The first anti-commutator vanishes as well, for
$\hat T\hat\beta'$ commutes with $T$ (by the first equation
\eqref{eq:27} and the order zero condition) while it anticommutes with
$[D,\alpha']T$. The latter assertion follows from the observation that
$\hat T\hat\beta'$ anticommutes with $[D,\alpha']$, since $\hat\beta'$
commutes with it (by the first order condition of the initial triple)
while $T$ anti-commutes with $[D,\alpha']$, as can be seen
from \eqref{eq:46b} noticing that
\begin{align}
\left[\left\{D, T\right\}, \hat\alpha\right]  &=DT\hat\alpha +
TD\hat\alpha - \hat\alpha DT - \hat\alpha TD,\\
&=[D,\hat\alpha]  T +
T[D,\hat\alpha] =\left\{ [D,\hat\alpha],  T\right\}.
\end{align}
So \eqref{eq:45} and \eqref{eq:41} (or equivalently \eqref{eq:46b})
are equivalent with the twisted first-order condition. Hence the result. \end{proof}

For an almost commutative geometry, these conditions are equivalent to
their restriction of the finite dimensional space.

\begin{corollary}
  \label{sec:twisted-first-order-1}
Consider the minimal twist of a real, almost commutative, geometry by
$T=\gamma_\M\otimes T_F$ such that the order zero condition
holds. Then the twisted first-order condition holds if and only if
\begin{equation}
  \label{eq:30fin}
\left\{\left\{D_F, T_F\right\}, J_F T_F J_F^{-1}\right\}=0 \quad\text{ and }\quad
[\left\{D_F,T_F\right\}, J_F \,\pi_F(m) \, J_F^{-1}]=0 \quad \forall
 m\in \A_F.
\end{equation}
\end{corollary}
\begin{proof}
For $D$ as in \eqref{eq:16} and $T=\gamma_\M\otimes T_F$, one has
\begin{equation}
  \left\{D, T\right\}=\left\{\gamma_\M\otimes D_F, \gamma_\M\otimes T_F\right\}= \I\otimes\left\{D_F, T_F\right\},
\end{equation}
  so the second equation \eqref{eq:30fo} is equivalent to
  $\left\{\left\{D_F, T_F\right\}, J_F T_F J_F^{-1}\right\}=0 $. The
  first equation \eqref{eq:30fo} reads
  \begin{equation}
   0=\left\{\I\otimes\left\{D_F, T_F\right\}, {\mathcal J}\gamma_\M{\mathcal
           J}^{-1}\otimes J_FT_F J_F^{-1}\right\}=  {\mathcal J}\gamma_\M{\mathcal
           J}^{-1}\otimes \left\{\left\{D_F, T_F\right\}, J_F T_F
           J_F^{-1}\right\}
  \end{equation}
which is equivalent to the first equation \eqref{eq:30fin}.
\end{proof}

These conditions are automatically satisfied if $T_F$ anticommutes
with $D_F$, that is if the twisting operator $T$ is a grading. But
this may not be the only possibility.

\section{Conclusion}
In recent time several modifiations of the framework of
noncommutative geometry have been proposed to get extra scalar fields,
beyond the Standard Model: removing the first order condition
\cite{Chamseddine:2013fk,Chamseddine:2013uq}, twisting the real
structure \cite{T.-Brzezinski:2016aa,Magee:2020aa,Dabrowski:2019ac,Dabrowski:2019ab},
introducing non-associativity \cite{Boyle:2019ab},working in the
framework of \emph{geometric background} \cite{Besnard:2019aa,Besnard:2020ac}
(see \cite{Chamseddine:2019aa} for a recent review). The relation of
some of these procedures with the minimal twist presented here have
been investigated in \cite{Martinetti:2021aa}  regarding the removal of the first order condition \cite{Martinetti:2021aa},
and in \cite{Brzezinski:2018aa}  regarding the twisting of the real
structure  (see also  \cite{Goffeng:2019aa}).

The main result of this note is that the twist-by-grading is not the
only possibility for minimally twisting the Standard Model. It is true
that the minimal twist of an almost commutative geometry by a twisting
operator of the form $\T\otimes T_F$ is a twisted spectral triple if
and only if $\T=\gamma_\M$. However,  $T_F$ does not need to be a grading of
the finite dimensional space. If one requires the order zero
condition, then 
\begin{equation}
  \label{eq:30conc}
[T_F, J_F T_F J_F^{-1}]=0 \quad\text{ and }\quad
[T_F, J_F \,\pi_F(m) \, J_F^{-1}]=0 \quad \forall
 m\in \A_F,
\end{equation}
and if one requires the twisted first-order condition, then 
\begin{equation}
  \label{eq:30finconc}
\left\{\left\{D_F, T_F\right\}, J_F T_F J_F^{-1}\right\}=0 \quad\text{ and }\quad
[\left\{D_F,T_F\right\}, J_F \,\pi_F(m) \, J_F^{-1}]=0 \quad \forall
 m\in \A_F.
\end{equation}
 These conditions hold in more generality, for the minimal twist of an
 arbitrary (real) spectral triple, as shown in Propositions
 \ref{sec:order-zero-condition-1} and \ref{sec:twisted-first-order-2}.
 Classifying all the solutions of
these constraints for the spectral triple of the Standard Model will be
the object of a future work.

 \bibliographystyle{abbrv}
  \bibliography{/Users/pierre/physique/articles/bibdesk/biblio}
 \end{document}